\documentclass[11pt,a4paper]{article}

\usepackage{graphicx}
\usepackage{amsmath, amsthm, amsfonts}

\newtheorem{prp}{Proposition}[section]
\usepackage[latin1]{inputenc}
\usepackage[symbol]{footmisc}

\title{Optimal retirement in presence of stochastic labor income: a free boundary approach in an incomplete market framework} 
\author{Daniele Marazzina\footnote{daniele.marazzina@polimi.it}\\ 
Politecnico di Milano, Department of Mathematics, Milano, Italy \\
}

\begin{document}

\maketitle

\begin{abstract}
In this work, we address the optimal retirement problem in the presence of a stochastic wage, formulated as a free boundary problem. Specifically, we explore an incomplete market setting where the wage cannot be perfectly hedged through investments in the risk-free and risky assets that characterize the financial market.
	
\end{abstract}

\section{Introduction}
In this article, we extend the work of \cite{BM}, where a stochastic wage perfectly correlated to the risky asset is considered, assuming that the wage may not be perfectly hedged exploiting the financial market, and we show how the free boundary problem is modified due to the incomplete market assumption, exploiting the theoretical framework in \cite{Tepla}.

More precisely, we study a stochastic control problem involving the consumption-portfolio-leisure
policy and the optimal stopping time of retirement. Relating to the leisure rate, the agent
earns the labour income with a stochastic wage rate. The sum of labour and leisure rates is assumed
to be constant. As the complement, the labour rate is lower bounded by a positive constant for the consideration of retaining the employment state before retirement, and then, after declaring retirement, the
agent realizes the full leisure rate. 

The problem is solved exploiting the duality method, which can be summarized into four steps. We
first tackle the post-stopping time problem to deduce a closed form of the corresponding value
function. Then we apply the Legendre-Fenchel transform to the utility function and the value
function of post-stopping time problem obtained in the first step. Afterwards, we construct the
duality between the optimal control problem with the individual's shadow prices problem, by
the aid of the liquidity and budget constraints and the dual transforms acquired before. Finally,
we cast the dual shadow price problem as a free boundary problem. We refer to \cite{BM,CSS,DM2,DM3,DL} for related works. The obtained free boundary problem can not be solved analytically, but only numerically, e.g., exploiting a finite difference scheme as in \cite{BM}. 

The article is organized as follows: in Section \ref{MOD} we deal with the considered problem, in Section \ref{DUAL} we show how the duality approach is applied, finally Section \ref{HJB} contains the main result of the article, that is the partial differential equation (pde) of the related free boundary problem.

\section{The Optimal Consumption-Investment Problem}\label{MOD}

We consider a continuous time economy. The infinite horizon life
of the agent is divided in two parts: before retirement and after
retirement. In the first part of his life the agent consumes,
invests in financial markets and chooses the labor supply rate earning an exogenous
stochastic income, after retirement he does not work, he only
consumes and invests his wealth in financial markets. After retirement the
agent cannot go back to work.

The labor income process $Y(t)$ (wage) is exogenous and stochastic. More specifically, we assume
\begin{equation}
\label{labor}
dY(t)= \mu_1 Y(t)dt + \boldsymbol{\mu}^{\top}_2 Y(t)  d\boldsymbol{B}(t), \
\ Y(0)=Y_0>0,
\end{equation}
where $\mu_1$ is a scalar function, $\boldsymbol{\mu_2}$ is a
$2d$-vector of functions, such that $Y(t)> 0,\ \forall t>0$ and $\boldsymbol{B}(t)$ is a $2d$-vector of
Brownian motions.

We assume that there are 2 assets: the risk-free asset with a constant
instantaneous interest rate $r$ and a risky assets. Asset price
$S(t)$ evolve as
\begin{equation}\label{risky}
d{S}(t)={S}(t) \left({b} dt+ \sigma d{B}_1(t)\right), \ \ S(0)=S_{0},
\end{equation}
where $\boldsymbol{b}$ is the constant drift of the risky
asset price, $\sigma$ is the volatility matrix. Notice that (\ref{risky}) can be re-written as 
\begin{equation}\label{risky2}
d{S}(t)={S}(t) \left({b} dt+ \boldsymbol{\sigma}^{\top} d\boldsymbol{B}(t)\right), \ \ S(0)=S_{0},
\end{equation}
with $\boldsymbol{\sigma}=[\sigma;0]$. Since the kernel of the vector $\boldsymbol{\sigma}$ is not empty, the market is incomplete: we are not able to perfectly hedge/insurance the wage just using the market assets. As in \cite{Tepla}, let us define the space 
\begin{eqnarray}
K(\sigma)=\Big{\{} \boldsymbol{v}\in \mathbb{R}^2\ : \ \boldsymbol{\sigma}^{\top}\boldsymbol{v}=0\Big{\}}.\label{K}
\end{eqnarray}

We denote by $c(t)$, $\theta(t)$ and
$l(t)$ consumption, risky asset portfolio and leisure
processes, respectively. We make the following assumptions on these processes:

\begin{itemize}
\item
$c(t)$ is a non negative progressively process measurable with respect to ${\cal F}$, such that $\int_0^t
c(s)ds < \infty, \  \forall t \ge 0$ a.s.,
\item
$\theta(t)$ represents the amount of wealth invested in the risky asset at time $t$, it is a progressively adapted process
with respect to ${\cal F}$, such that $\int_0^t \theta^2(s) ds < \infty, \  \forall t \ge 0$ a.s.,
\item
$l(t) \in\left[0,1\right]$ is the rate of leisure at time $t$, $l(t)$ is a
process measurable with respect to ${\cal F}$.
\end{itemize}

Let $\tau \ge 0$ be the retirement date chosen by the
agent, i.e., $\tau \in {\cal S}$, where ${\cal S}$ denotes the set of all ${\cal F}$-stopping time.
The retirement decision is irreversible:
during his working life the agent chooses the leisure rate, after
retirement he fully enjoys leisure. At time $t$ the agent is endowed with
one unit of ``time'', before retirement he allocates his time between labor ($1-l(t)$) and leisure
($l(t)$) with the constraint that he has to work at least $1-L$, after retirement
he only enjoys leisure:
\begin{equation*}
0\leq l(t) \leq L<1 \ \textrm{if}\ 0 \leq t < \tau, \ \ \ l(t)=1 \
\textrm{if}\ t\geq \tau.
\end{equation*}

Let $W(0)=W_0$ be the initial wealth. The
consumption-investment-leisure strategy $(c,l,\theta)$ is admissible if it
satisfies the above technical conditions and the wealth process
$W(t)$ satisfies the dynamic budget constraint
\begin{equation}
\label{budget}
dW(t)\!= \!\left[(1-l(t))Y(t)\! -\! c(t)\right]dt +
{\theta}(t)\left({b}\,dt+\boldsymbol{\sigma}^{\top}d\boldsymbol{B}(t) \right) +
\left(\!W(t)-\!\theta(t)\!\right)\! r dt
\end{equation}
where $(1-l(t))Y(t) dt$ is the labor income received before
retirement, i.e., for $0 \le t < \tau$,
${\theta}(t)\left({b}\,dt+\boldsymbol{\sigma}^{\top}d\boldsymbol{B}(t) \right)$ refers to
the financial wealth evolution and $\left(W(t)-\theta(t)\right) r dt$ to
the wealth invested in the money market (risk-free bond).

We impose a no borrowing condition
\begin{equation}
\label{wealth}
W(t)\geq 0, \ \ \forall t \ge 0
\end{equation}
that is the agent cannot borrow money in the risk-free market
against future income.

The agent maximizes the expected utility over the infinite horizon:
preferences are time separable with exponential discounting.
The instantaneous utility at time $t$ is a function of consumption
and leisure, i.e., $u(c(t),l(t))$. Therefore, given the
initial wealth $W_0$ and the labor income process $Y(t)$, we look
for an admissible process triplet $(c,l,{\theta})$ and a stopping time $\tau \in {\cal S}$
that maximize
\begin{equation}\label{exp}
E\left[ \int_{0}^{+\infty} e^{-\beta t}u(c(t), l(t)) dt \right],
\end{equation}
subject to the budget constraint (\ref{budget}) and the no
borrowing condition (\ref{wealth}). $\beta$ is the subjective
discount rate and the utility function is given by
\begin{equation}\label{utilityfunction}
u(c,l)=\frac{1}{\alpha}\frac{(l^{1-\alpha}c^{\alpha})^{1-\gamma}}{1-\gamma}, \ \gamma>0, \  0<\alpha<1.
\end{equation}

The infinite horizon is a strong
assumption that allows us to handle analytically the problem; to
make the setting more plausible, we may allow for an hazard rate of
mortality as in \cite{DL} adding a positive coefficient to the discount
factor $\beta$.

Given the initial wealth $W_0$ and the labor income
$Y_0$, the agent maximizes the expected utility
(\ref{exp}):
\begin{eqnarray*}
{\cal J}(W_0,Y_0;c,l,{\theta},\tau)&:=&E\left[ \int_{0}^{+\infty} e^{-\beta t}u(c(t), l(t)) dt \right]\\
&=&E\left[ \int_{0}^{\tau}e^{-\beta t}u(c(t), l(t)) dt
+\int_{\tau}^{+\infty} e^{-\beta t}u(c(t), 1) dt\right]
\end{eqnarray*}
acting on $(c,l,{\theta},\tau)$ with $\tau \in {\cal S}$ and subject to the budget constraint
(\ref{budget}) and the no borrowing condition (\ref{wealth}). The value function is defined as
\begin{equation}\label{value}
{\cal V}(W_0,Y_0):=\sup_{(c, l, {\theta}, \tau)\in {\cal A}} {\cal
J}(W_0,Y_0;c,l,{\theta}, \tau),
\end{equation}
where ${\cal A}$ denotes the set of admissible strategies,
i.e., processes ($c,l,{\theta}, \tau$) that satisfy (\ref{budget})
and (\ref{wealth}), $\tau \in {\cal S}$, such that ${\cal
J}(W_0,Y_0;c,l,{\theta},\tau)<+\infty$.

\section{A dual approach}\label{DUAL}

As in \cite{Tepla}, considering (\ref{K}), let us define the market price of risk 
\begin{equation*}
\boldsymbol{\Theta}(\boldsymbol{v})\!:=\!\boldsymbol{\sigma}\left(\boldsymbol{\sigma}^{\top}\boldsymbol{\sigma}\right)^{-1}(b -r)+\boldsymbol{v},\	 \forall \boldsymbol{v}\in K(\sigma),
\end{equation*}
the state-price-density process (or pricing kernel)
\begin{equation*}
H(t,\boldsymbol{v}):=e^{-\left(r+\frac{1}{2}|\boldsymbol{\Theta}(\boldsymbol{v})|^2\right)t-\boldsymbol{\Theta}^{\top}(\boldsymbol{v})
\boldsymbol{B}(t)}
\end{equation*}
and, for a given $t>0$, the equivalent risk neutral martingale measure
\begin{equation*}
\overline{P}(A):=E\left[e^{-\boldsymbol{\Theta}^{\top}(\boldsymbol{v})
\boldsymbol{B}(t) - \frac{1}{2}|\boldsymbol{\Theta}(\boldsymbol{v})|^2 t}
\mathbf{1}_{A}\right] \ \  \forall A \in {\cal F}_t.
\end{equation*}
Under the (non-unique) risk neutral density we can redefine the budget and the
no borrowing constraint as static constraints. 
Following \cite{CSS}, by the optional
sampling theorem we have that the budget constraint becomes
\begin{equation}\label{BC}
E\left[\int_{0}^{\tau}H(t,\boldsymbol{v})\left(c(t)+Y(t) (l(t)-1)\right)dt +
H(\tau,\boldsymbol{v})W(\tau)\right]\leq W_0,\ \ \forall \boldsymbol{v}\in K(\sigma)
\end{equation}
and the no borrowing constraint ($W(t)\geq 0, \ \forall t,  \ 0 \le t \le \tau$) is equivalent to
\begin{equation}\label{LC}
E_{t}\left[\int_{t}^{\tau}\frac{H(s,\boldsymbol{v})}{H(t,\boldsymbol{v})}\left(c(s)+Y(s)
(l(s)-1)\right)ds + \frac{H(\tau,\boldsymbol{v})}{H(t,\boldsymbol{v})}W(\tau) \right]\geq 0
\end{equation}
$\forall \ t\ :\ 0\leq t\leq \tau,\ \forall \boldsymbol{v}\in K(\sigma).$ Note that the no borrowing constraint is never binding after retirement.
The objective function can be rewritten as
\begin{eqnarray}
{\cal J}(W_0,Y_0;c,l,{\theta},\tau)\!\!\!\!&=&\!\!\!\!
E\left[ \int_{0}^{\tau}\! e^{-\beta t} u(c(t), l(t)) dt + e^{-\beta \tau}\int_{\tau}^{+\infty}\!\! e^{-\beta (t-\tau)} u(c(t), 1) dt\right]\nonumber\\
&=&\!\!\!\!  E\left[ \int_{0}^{\tau} e^{-\beta t} u(c(t), l(t)) dt +
e^{-\beta \tau}U(W(\tau))\right],\nonumber
\end{eqnarray}
where
\begin{eqnarray}\label{EU}
U(W(\tau)):=\sup_{(c,{\theta})} E\left[ \int_{\tau}^{+\infty}
e^{-\beta (t-\tau)} u(c(t), 1) dt\right]
\end{eqnarray}
is the optimal expected utility attainable at time $\tau$ with wealth $W(\tau)$:
for $t \ge \tau$ the agent solves the classical optimal consumption-portfolio problem
with initial wealth $W(\tau)$ subject to the budget constraint.
$U(W(\tau))$ is the indirect utility function associated with the maximization problem at time $\tau$.

We define the convex conjugate of the
utility function $u(c,l)$:
\begin{equation}
\label{dual}
\widetilde{u}(z,Y):=\max_{\scriptsize{\begin{matrix}c\geq0\\ 0\leq
l\leq L\end{matrix}}} u(c,l)-(c+Yl)z.
\end{equation}
In a similar way, we also define the convex conjugate of $U$:
\begin{eqnarray}\label{dualU}
\widetilde{U}(z):=\sup_{w\geq 0} U(w)-wz.
\end{eqnarray}
If $I$ denotes the inverse of $U'$, i.e., $I(z)=\inf\{w\,:\, U'(w)=z \}$, then
we also have $\widetilde{U}(z)=U(I(z))-zI(z)$. 

Let $\widehat{c}$ and $\widehat{l}$ be the pair of
processes that provides a solution to (\ref{dual}),
then the following proposition holds true.

\begin{prp}\label{prp:utilde}
Let
\begin{eqnarray}\label{eq:xtilde}
\widetilde{z}=\left(\frac{\alpha}{1-\alpha}Y\right)^{\alpha(1-\gamma)-1}L^{-\gamma}.
\end{eqnarray}
If $z\geq\widetilde{z}$, then
\begin{equation}
\widetilde{u}(z,Y)= \frac{1}{\alpha}\frac{(\widehat{l}^{\,1-\alpha}\ \widehat{c}^{\,\alpha})^{1-\gamma}}{1-\gamma}-(\widehat{c}+Y\widehat{l})z,\label{utilde1}
\end{equation}
where $\widehat{c}=\frac{\alpha}{1-\alpha}Y\widehat{l}\ \ \textrm{and} \ \
\widehat{l}=z^{-\frac{1}{\gamma}}
\left(\frac{\alpha}{1-\alpha}Y\right)^{\frac{\alpha(1-\gamma)-1}{\gamma}}.$\\
If $z<\widetilde{z}$, then
\begin{equation}
\widetilde{u}(z,Y)= \frac{\left(\widehat{c}^{\,\alpha}L^{1-\alpha}\right)^{(1-\gamma)}}{\alpha(1-\gamma)}-(\widehat{c}+YL)z,\label{utilde2}
\end{equation}
where $\widehat{c}=\left(zL^{(\alpha-1)(1-\gamma)}\right)^{\frac{1}{\alpha(1-\gamma)-1}}.$
\end{prp}

\begin{proof}
See \cite[Proposition 3.1]{BM}.
\end{proof}

Let $\lambda>0$ be a Lagrange Multiplier. We consider a non-increasing process $D(t)>0$ with $D(0)=1$ to take into account
the no bankruptcy constraint, as in \cite{CSS,HP}, i.e., $D(t)$ is the integral of the shadow price of the no bankruptcy constraint. Then, the following proposition holds.
\begin{prp}
\label{JV}
Let
\begin{eqnarray}
\widetilde{V}(\lambda,D,\tau,Y_0,\boldsymbol{v})\!\!\!\!&:=&\!\!\!\!E\Big{[} \int_{0}^{\tau}e^{-\beta t}\left( \widetilde{u}\left(\lambda D(t) e^{\beta t}H(t,\boldsymbol{v}),\,Y(t)\right)+\lambda Y(t) D(t) e^{\beta t}H(t,\boldsymbol{v})\right) dt \nonumber\\
&+&\!\!\!\! e^{-\beta \tau}\widetilde{U}\left(\lambda D(\tau)e^{\beta\tau}H(\tau,\boldsymbol{v}) \right) \Big{]}, \label{Vldty}
\end{eqnarray}
then
\begin{equation}
{\cal J}(W_0,Y_0;c,l,\boldsymbol{\theta},\tau)\leq
\widetilde{V}(\lambda,D,\tau,Y_0,\boldsymbol{v}) + \lambda W_0.
\label{ineq}
\end{equation}
\end{prp}

\begin{proof}
Considering (\ref{dual}) and (\ref{dualU}), it holds
\begin{eqnarray}
&&\!\!\!\!\!\!\!\!\!\!{\cal J}(W_0,Y_0;c,l,\boldsymbol{\theta},\tau)\!=\! E\Big{[}\! \int_{0}^{\tau}\!\!\!\!e^{-\beta t}\!\left\{u(c(t),l(t))\!-\!\lambda D(t) e^{\beta t}H(t,\boldsymbol{v})(c(t)\!+\!Y(t)l(t)) \right\}dt\nonumber\\
&&+e^{-\beta \tau}\left\{ U(W(\tau))-\lambda D(\tau)e^{\beta \tau}H(\tau,\boldsymbol{v})W(\tau)\right\} \Big{]}\nonumber\\
&&+\lambda\, E\left[ \int_{0}^{\tau}D(t) H(t,\boldsymbol{v})(c(t)+Y(t) l(t)) dt+D(\tau)H(\tau,\boldsymbol{v}) W(\tau) \right]\nonumber\\
&&\leq E\left[ \int_{0}^{\tau}e^{-\beta t}\,\widetilde{u}\left(\lambda D(t) e^{\beta t}H(t,\boldsymbol{v}),\,Y(t)\right)dt+e^{-\beta \tau} \widetilde{U}\left(\lambda D(\tau)e^{\beta \tau}H(\tau,\boldsymbol{v})\right)\right]\nonumber\\
&&+ \lambda\, E\left[ \int_{0}^{\tau}D(t) H(t,\boldsymbol{v})(c(t)+Y(t) l(t))
dt+D(\tau)H(\tau,\boldsymbol{v}) W(\tau) \right]\label{j2}.
\end{eqnarray}
Exploiting the constraints (\ref{BC}), (\ref{LC}) and being $D(t)$ non increasing such that $D(0)=1$ we obtain
\begin{eqnarray}
&&\!\!\!\!E\!\left[ \int_{0}^{\tau}\!\!D(t) H(t,\boldsymbol{v})(c(t)\!+\!Y(t) l(t)) dt\!+\!D(\tau)H(\tau,\boldsymbol{v}) W(\tau)\! \right]\nonumber\\
&&\!\!\!\!\!\!\!\!\!\!=E\!\left[ \int_{0}^{\tau}\!\!D(t) H(t,\boldsymbol{v})(c(t)\!+\!(l(t)\!-\!1) Y(t)) dt\!+\!D(\tau)H(\tau,\boldsymbol{v}) W(\tau)\!+\!\int_{0}^{\tau}\!\!Y(t) D(t)H(t,\boldsymbol{v})dt\! \right] \nonumber\\
&&\!\!\!\!\!\!\!\!\!\!= E\!\left[ \int_{0}^{\tau}\!\!Y(t) D(t)H(t,\boldsymbol{v})dt\! \right] + E\!\left[\! H(\tau,\boldsymbol{v}) W(\tau)\! +\! \int_{0}^{\tau}\!\!H(t,\boldsymbol{v})(c(t)\!+\!(l(t)\!-\!1) Y(t))dt\!\right]\nonumber\\
&&\!\!\!\!\!\!\!\!\!\!+ E\!\left[\int_{0}^{\tau}\!\!H(t,\boldsymbol{v})E_{t}\!\left[\!\frac{H(\tau,\boldsymbol{v})}{H(t,\boldsymbol{v})}W(\tau)\!+\!\int_t^{\tau}\!\!\frac{H(s,\boldsymbol{v})}{H(t,\boldsymbol{v})}(c(s) \!+\!(l(s)\!-\!1) Y(s))ds \right]dD(t) \right]\nonumber\\
&&\!\!\!\!\!\!\!\!\!\!\leq E\!\left[ \int_{0}^{\tau}\!\!Y(t) D(t)H(t,\boldsymbol{v})dt\! \right] \! +\!  W_0.\label{leq2}
\end{eqnarray}
Thus inequality (\ref{ineq}) holds true.
\end{proof}

For a fixed stopping time $\tau \in {\cal S}$ and initial wealth $W_0$, we
denote by ${{\cal A}_{\tau}}$ the set of admissible processes
$(c,\,l,\,\boldsymbol{\theta})$. Thus we define
\begin{equation}\label{valuet}
V_{\tau} (W_0,Y_0):=\sup_{(c,\,l,\,\boldsymbol{\theta})\in{\cal
A}_{\tau}} {\cal J}(W_0,Y_0;c,l,\boldsymbol{\theta},\tau).
\end{equation}
From Proposition \ref{JV}, we obtain
\begin{equation}\label{dis}
V_{\tau}(W_0,Y_0) \leq \inf_{\lambda>0,\,D(t)>0,\, \boldsymbol{v}\in K(\sigma)}\left[
\widetilde{V}(\lambda,D,\tau,Y_0,\boldsymbol{v}) + \lambda W_0\right].
\end{equation}
and the following result holds true.

\begin{prp}\label{JV-eq}
Equality in (\ref{dis}) holds if all the following conditions are satisfied:\\
- for $0\leq t<\tau$
\begin{eqnarray}\label{P4a}
\frac{\partial u}{\partial c}(c(t),l(t))=\lambda D(t) e^{\beta t}H(t,\boldsymbol{v})\ \ \textrm{ and}\ \ \frac{\partial u}{\partial l}(c(t),l(t))=\lambda Y(t) D(t) e^{\beta t}H(t,\boldsymbol{v})
\end{eqnarray}
if $l(t)\leq L$, otherwise
\begin{eqnarray}\label{P4b}
\frac{\partial u}{\partial c}(c(t),L)=\lambda D(t) e^{\beta t}H(t,\boldsymbol{v})\ \ \textrm{ and}\ \ l(t)=L;
\end{eqnarray}
- it holds
\begin{eqnarray*}
&&W(\tau)=I(\lambda D(\tau) e^{\beta \tau}H(\tau,\boldsymbol{v})),\\
&&E\left[\int_{0}^{\tau}H(s,\boldsymbol{v})(c(s)+(l(s)-1)Y(s))ds+H(\tau,\boldsymbol{v})W(\tau)
\right]=W_0;
\end{eqnarray*}
- for any $t\in [0,\tau)$ such that $D(t)$ is not constant, i.e., $dD(t) \neq 0$,
\begin{eqnarray}
E_t\left[\int_{t}^{\tau}\frac{H(s,\boldsymbol{v})}{H(t,\boldsymbol{v})}(c(s)+(l(s)-1)Y(s))ds+\frac{H(\tau,\boldsymbol{v})}{H(t,\boldsymbol{v})}W(\tau)
\right]=0.\label{prp2}
\end{eqnarray}
\end{prp}

\begin{proof}
Equality in (\ref{dis}) holds if and only if equalities in (\ref{j2}) and (\ref{leq2}) hold true.\\
Let us start with (\ref{j2}): this equality holds if
\begin{eqnarray*}
u(c(t),l(t))-\lambda D(t) e^{\beta t}H(t,\boldsymbol{v})(c(t)+Y(t)l(t))&=&\widetilde{u}\left(\lambda D(t) e^{\beta t}H(t,\boldsymbol{v}),\,Y(t)\right)\\
U(W(\tau))-\lambda D(\tau)e^{\beta \tau}H(\tau,\boldsymbol{v})W(\tau) &=&\widetilde{U}\left(\lambda D(\tau)e^{\beta \tau}H(\tau,\boldsymbol{v})\right).
\end{eqnarray*}
By Proposition \ref{prp:utilde}, this is equivalent to set for $0\leq t<\tau$ 
\begin{eqnarray*}
\frac{\partial u}{\partial c}(c(t),l(t))=\lambda D(t) e^{\beta t}H(t,\boldsymbol{v})\ \ \textrm{and}\ \  \frac{\partial u}{\partial l}(c(t),l(t))=\lambda Y(t) D(t) e^{\beta t}H(t,\boldsymbol{v})
\end{eqnarray*}
if $l(t)\leq L$, otherwise
\begin{eqnarray*}
\frac{\partial u}{\partial c}(c(t),L)=\lambda D(t) e^{\beta t}H(t,\boldsymbol{v})\ \ \textrm{and}\ \ l(t)=L;
\end{eqnarray*}
and, reasoning as above for the utility function $U$, to impose
\begin{eqnarray*}
U'(W(\tau))=\lambda D(\tau) e^{\beta \tau}H(\tau,\boldsymbol{v}),
\end{eqnarray*}
i.e., $W(\tau)=I(\lambda D(\tau) e^{\beta \tau}H(\tau,\boldsymbol{v}))$.\\
Now we deal with inequality (\ref{leq2}): it becomes an equality if
\begin{eqnarray}
&&\!\!\!\!\!\!\!\!\!\!\!\!\!\!\!\!\!\!\!\!\!\!\!\!E\!\left[H(\tau,\boldsymbol{v}) W(\tau) + \int_{0}^{\tau}H(t,\boldsymbol{v})(c(t)+(l(t)-1) Y(t))dt\right]=W_0, \label{proof1}\\
&&\!\!\!\!\!\!\!\!\!\!\!\!\!\!\!\!\!\!\!\!\!\!\!\!E\!\left[\int_{0}^{\tau}\!\!\!\!H(t,\boldsymbol{v})E_{t}\!\!\left[\frac{H(\tau,\boldsymbol{v})}{H(t,\boldsymbol{v})}W(\tau)\!+\!\int_t^{\tau}\!\!\frac{H(s,\boldsymbol{v})}{H(t,\boldsymbol{v})}(c(s) +(l(s)-1) Y(s))ds \right]dD(t) \right]=0,\label{proof2}
\end{eqnarray}
and the last condition is equivalent to (\ref{prp2}).\\
Notice that (\ref{proof1}) and (\ref{proof2}) imply (\ref{BC}) and (\ref{LC}), respectively. See \cite{CSS} for further details.
\end{proof}

Optimality conditions in (\ref{P4a}-\ref{P4b}) provide a relation between the
optimal consumption, the leisure process and the triple ($\lambda, D(t),\,v$).
As described in \cite[Section 3.2]{HPea}, we deal with the dual Arrow-Debreu state price problem, which corresponds to minimize over $\boldsymbol{v}\in K(\sigma)$ the maximum expected utility that the agent could obtain in a complete market with state price density $\Theta(\boldsymbol{v})$, defining the minimax local martingle measure. Thus $\boldsymbol{v}^*$ is chosen in such a way to mimimize the agent wealth. Once $\boldsymbol{v}^*$, the optimal multiplier $\lambda^*$ and the non-increasing process $D^*(t)$ are computed, we can use these conditions to obtain the optimal processes $c^*(t)$ and $l^*(t)$, as discussed in \cite[Appendix]{BM}.

In order to compute the triple $(\lambda^*,D^*(t),\boldsymbol{v}^*)$, we proceed as follows. Let
\begin{eqnarray*}
\widetilde{V}(\lambda, Y_0)&:=&\sup_{\tau\in[0,+\infty]}\inf_{\boldsymbol{v}\in K(\sigma),\,D(t)>0}
\widetilde{V}(\lambda,D,\tau, Y_0, \boldsymbol{v})\\
&=&\sup_{\tau\in[0,+\infty]}\inf_{D(t)>0}
\widetilde{V}(\lambda,D,\tau, Y_0, \boldsymbol{v}^*).
\end{eqnarray*}
Since (\ref{value}) and (\ref{valuet}) imply
\begin{equation*}
{\cal V}(W_0,Y_0)=\sup_{\tau\in[0,+\infty]}V_{\tau}(W_0,Y_0),
\end{equation*}
then, under the conditions of Proposition \ref{JV-eq}, it holds
\begin{eqnarray}\label{P2}
{\cal V}(W_0,Y_0)&=&\sup_{\tau\in[0,+\infty]}\inf_{\lambda>0,\,D(t)>0}\left[
\widetilde{V}(\lambda,D,\tau,Y_0,\boldsymbol{v}^*) + \lambda W_0\right]\\
&=&\inf_{\lambda>0}\left[
\widetilde{V}(\lambda,Y_0) + \lambda W_0\right].\nonumber
\end{eqnarray}
Notice that this formulation is related to \cite[Problem (D') page 277]{HPea}.

Let us define the process
\begin{equation}\label{P3}
z(t,\boldsymbol{v})=\lambda D(t) e^{\beta t} H(t,\boldsymbol{v}),\ \  z(0,\boldsymbol{v})=\lambda
\end{equation}
and
\begin{eqnarray}\nonumber
\phi(t,z,y):=\sup_{\tau>t}\inf_{\boldsymbol{v}\in K(\sigma),\,D(t)>0}\!\!
E\!\Big{[}\!\int_{t}^{\tau}\!\!\!\!e^{-\beta
s}\!\left\{\widetilde{u}(z(s,\boldsymbol{v}),Y(s))\!+\!Y(s) z(s,\boldsymbol{v})\right\}ds \!+\!
e^{-\beta\tau}\widetilde{U}(z(\tau,\boldsymbol{v}))\\
\Big{|} z(t,\boldsymbol{v})=z, Y(t)=y \Big{]}.\label{phi}
\end{eqnarray}

Since $\phi(0,\lambda,Y_0)=\widetilde{V}(\lambda,Y_0)$, once $\phi$ is computed, we can use (\ref{P2}) to obtain the value function ${\cal V}(W_0,Y_0)$ and thus the optimal couple $(\lambda^*,D^*(t))$ and the optimal strategies, as shown in \cite[Appendix]{BM}.

\section{The Bellman equation and the free boundary problem}
\label{HJB}
To address the optimization problem and to find out $\phi(\cdot)$ and the function $\boldsymbol{v}^*$,
we start from the dual function $\widetilde{U}$.
After retirement the optimal consumption-investment problem is a classical
Merton problem, considering the utility function (\ref{utilityfunction}), we have
\begin{equation*}
U(W(\tau))=\sup_{(c,\boldsymbol{\theta})} E\left[ \int_{\tau}^{+\infty}
e^{-\beta (t-\tau)} \frac{c^{\alpha(1-\gamma)}}{\alpha(1-\gamma)}
dt\right],
\end{equation*}
subject to the dynamic budget constraint (\ref{budget}) without labor income, i.e., $l(t)=1\, \forall t \geq \tau$.

It follows from standard results (see \cite[Chapter 3]{KS}) that
\begin{equation*}
U(w)=\left(\frac{1}{\xi}\right)^{\Gamma}\frac{w^{1-\Gamma}}{1-\Gamma}
\end{equation*}
where $\Gamma=1-\alpha(1-\gamma)$ and $\xi=\frac{\Gamma-1}{\Gamma}\left(r+\frac{\boldsymbol{\Theta}^2(\boldsymbol{v}^*)}{2\Gamma}
\right)+\frac{\beta}{\Gamma}$. In order to guarantee that the function $U$ is well defined, we assume $\xi>0$. Thus, since $U$ is known analytically, we can also compute $\widetilde{U}$ in (\ref{dualU}):
\begin{equation*}
\widetilde{U}(z)=\frac{\Gamma}{\xi(1-\Gamma)}z^{ \frac{\Gamma-1}{\Gamma}}.
\end{equation*}

Since by It\^o's formula we have, for any $\boldsymbol{v} \in K(\sigma)$
\begin{eqnarray*}
\frac{dz(t,\boldsymbol{v})}{z(t,\boldsymbol{v})}=\frac{dD(t)}{D(t)}+(\beta-r)dt-\boldsymbol{\Theta}^{\top}(\boldsymbol{v}) d\boldsymbol{B}(t),\ \ \ z(0)=\lambda,
\end{eqnarray*}
in order to derive the Bellman equation, we follow \cite{HP} controlling the decreasing process $D$. Assume that $D$ is absolutely continuous with respect to $t$, i.e., there exists $\psi\geq0$ such that $dD(t)=-D(t)\psi(t)dt$. If $\phi$ is twice differentiable with respect to $z$ and $y$, then the Bellman equation becomes
\begin{eqnarray}\label{BE1}
\min_{\boldsymbol{v}\in K(\sigma),\, \psi\geq 0}\left\{ e^{-\beta t} (\widetilde{u}\left(z,\,y\right)+zy)-\psi(t)z\frac{\partial \phi}{\partial z}+\mathcal{L}_t^0\phi  \right\}=0,
\end{eqnarray}
where
\begin{equation*}
{\cal L}_t^0 \phi=\frac{\partial \phi}{\partial t}+(\beta-r)z\frac{\partial\phi}{\partial z}+
\mu_1 y\frac{\partial\phi}{\partial y}+\frac{1}{2}\boldsymbol{\Theta}^2(\boldsymbol{v})
z^2\frac{\partial^2\phi}{\partial z^2}+\frac{1}{2}
\boldsymbol{\mu}_2^2 y^2\frac{\partial^2\phi}{\partial y^2}- \boldsymbol{\Theta}^{\top}(\boldsymbol{v})
\boldsymbol{\mu}_2 zy\frac{\partial^2\phi}{\partial z\partial y}
\end{equation*}
is the differential operator for the dual indirect utility function. Condition (\ref{BE1}) can be rewritten compactly as
\begin{eqnarray}\label{BE2}
\min\left\{ e^{-\beta t} (\widetilde{u}\left(z,\,y\right)+zy)+\mathcal{L}_t\phi,\ -\frac{\partial \phi}{\partial z}  \right\}=0,
\end{eqnarray}
where
\begin{equation*}
{\cal L}_t \phi=\frac{\partial \phi}{\partial t}+\mu_1y\frac{\partial\phi}{\partial y}+\frac{1}{2}
\boldsymbol{\mu}_2^2 y^2\frac{\partial^2\phi}{\partial y^2}+ 
(\beta-r)z\frac{\partial\phi}{\partial z}
+ \min_{\boldsymbol{v}\in K(\sigma)} \left[ \frac{1}{2}\boldsymbol{\Theta}^2(\boldsymbol{v})
z^2\frac{\partial^2\phi}{\partial z^2}- \boldsymbol{\Theta}^{\top}(\boldsymbol{v})
\boldsymbol{\mu}_2 z y\frac{\partial^2\phi}{\partial z\partial y}\right]
\end{equation*}
is the differential operator for the dual indirect utility function, i.e.,
\begin{equation}\label{operator}
{\cal L}_t \phi=\frac{\partial \phi}{\partial t}+\mu_1y\frac{\partial\phi}{\partial y}+\frac{1}{2}
\boldsymbol{\mu}_2^2 y^2\frac{\partial^2\phi}{\partial y^2}+ 
(\beta-r)z\frac{\partial\phi}{\partial z}
+ \frac{1}{2}\boldsymbol{\Theta}^2(\boldsymbol{v}^*)
z^2\frac{\partial^2\phi}{\partial z^2}- \boldsymbol{\Theta}^{\top}(\boldsymbol{v}^*)
\boldsymbol{\mu}_2 z y\frac{\partial^2\phi}{\partial z\partial y}
\end{equation}
and 
\begin{equation*}
\boldsymbol{v}^*= arg \min_{\boldsymbol{v}\in K(\sigma)} 
\frac{1}{2}\boldsymbol{\Theta}^2(\boldsymbol{v})
z^2\frac{\partial^2\phi}{\partial z^2}- \boldsymbol{\Theta}^{\top}(\boldsymbol{v})
\boldsymbol{\mu}_2 z y\frac{\partial^2\phi}{\partial z\partial y}.
\end{equation*}
Notice that the minimum over $\boldsymbol{v}\in K(\sigma)$ comes from the minimax local martingale measure \cite[Page 275]{HPea}. Moreover, since $\Theta(\boldsymbol{v})$ is a monotonic function with respect to $\boldsymbol{v}$, it holds
\begin{equation*}
{\cal L}_t \phi=\frac{\partial \phi}{\partial t}+\mu_1y\frac{\partial\phi}{\partial y}+\frac{1}{2}
\boldsymbol{\mu}_2^2 y^2\frac{\partial^2\phi}{\partial y^2}+ 
(\beta-r)z\frac{\partial\phi}{\partial z}
+ \min_{\Theta(\boldsymbol{v}),\,\boldsymbol{v}\in K(\sigma)} \left[\frac{1}{2}\boldsymbol{\Theta}^2(\boldsymbol{v})
z^2\frac{\partial^2\phi}{\partial z^2}- \boldsymbol{\Theta}^{\top}(\boldsymbol{v})
\boldsymbol{\mu}_2 z y\frac{\partial^2\phi}{\partial z\partial y}\right]
\end{equation*}
is the differential operator for the dual indirect utility function. Furthermore, we recall that the space $K(\sigma)$ consists of vectors $[0;x],\,x\in\mathbb{R}$, thus it holds
\begin{equation}\label{Theta}
[\boldsymbol{\Theta}]_1(\boldsymbol{v}^*)=\frac{b-r}{\sigma}\ \ \ \   [\boldsymbol{\Theta}]_2(\boldsymbol{v}^*)=\frac{[\boldsymbol{\mu}_2]_2 y\frac{\partial^2\phi}{\partial z\partial y}}{z\frac{\partial^2\phi}{\partial z^2}}=[\boldsymbol{v}^*]_2,
\end{equation}
where we denote with $[\cdot]_i$ the $i$-th component of a vector.

\subsection{Correlated Wage}\label{Cor}
If the wage is given by
\begin{equation*}
dY(t)= M_1 Y(t)dt + \rho M_2 Y(t)  d B_1(t)+  \sqrt{1-\rho^2}M_2 Y(t)  dB_2(t),
\end{equation*}
i.e., $\boldsymbol{\mu_2}$ is equal to the vector $[\rho M_2,\, \sqrt{1-\rho^2} M_2 ]$, $\rho$ being the correlation between the wage and the risky asset, then the above pde becomes
\begin{eqnarray*}
{\cal L}_t \phi&=&\frac{\partial \phi}{\partial t}+M_1y\frac{\partial\phi}{\partial y}+\frac{1}{2}
M_2^2 y^2\frac{\partial^2\phi}{\partial y^2}+ 
(\beta-r)z\frac{\partial\phi}{\partial z}
+ \frac{1}{2}\Theta^2
z^2\frac{\partial^2\phi}{\partial z^2}-\rho \Theta M_2 z y\frac{\partial^2\phi}{\partial z\partial y}
\\&+& \frac{1}{2} (\Theta^*)^2
z^2\frac{\partial^2\phi}{\partial z^2}- \Theta^*
\sqrt{1-\rho^2} M_2 z y\frac{\partial^2\phi}{\partial z\partial y},
\end{eqnarray*}
where
\begin{equation*}
\Theta=\sigma^{-1}(b-r), \ \ \textrm{and}\ \  \Theta^*=\frac{\sqrt{1-\rho^2} M_2 y\frac{\partial^2\phi}{\partial z\partial y}}{z\frac{\partial^2\phi}{\partial z^2}},
\end{equation*}
i.e.,
\begin{equation*}
{\cal L}_t \phi=\frac{\partial \phi}{\partial t}+M_1y\frac{\partial\phi}{\partial y}+\frac{1}{2}
M_2^2 y^2\frac{\partial^2\phi}{\partial y^2}+ 
(\beta-r)z\frac{\partial\phi}{\partial z}
+ \frac{1}{2}\Theta^2
z^2\frac{\partial^2\phi}{\partial z^2}-\rho \Theta M_2 z y\frac{\partial^2\phi}{\partial z\partial y}
-\frac{1}{2} (1-\rho^2) M_2^2 y^2 \frac{\left(\frac{\partial^2\phi}{\partial z\partial y}\right)^2}{\frac{\partial^2\phi}{\partial z^2}}.
\end{equation*}
Notice that if $\rho=\pm 1$ we move back to the complete market case, and, if $\rho=1$, we get the pde in \cite[Section 5]{BM}.

\section{Conclusions}

In this study, we have analyzed the optimal retirement problem in the presence of a stochastic labor income under an incomplete market framework. By extending the work of \cite{BM}, we considered a scenario where labor income cannot be perfectly hedged through financial market investments, leading to the formulation of a free boundary problem.

We approached the problem using the duality method, which allowed us to reformulate the consumption-investment-leisure optimization problem as a dual shadow price problem. This reformulation enabled us to derive a pde characterizing the free boundary, which determines the optimal retirement decision. 

The obtained results provide insights into the effects of labor income risk on the optimal retirement decision. In particular the market incompleteness modifies the optimal retirement boundary, affecting the optimal consumption, investment, and leisure policies. The duality approach offers a tractable framework to handle the complexities introduced by non-hedgeable labor income. Our formulation generalizes existing models, recovering the complete market case as a special scenario.

Our findings contribute to the understanding of optimal retirement decisions in stochastic environments and provide a framework for future studies on lifetime consumption-investment problems under market imperfections.

\end{document}